\documentclass[runningheads,a4paper]{llncs}

\usepackage{amssymb}
\usepackage{graphicx}

\usepackage{amssymb}
\usepackage{amsmath}
\usepackage{graphicx}
\usepackage{enumerate}
\graphicspath{{./figs/}}
\usepackage{float}
\usepackage{cite}
\usepackage{sgame}
\usepackage{setspace}
\usepackage{booktabs}
\usepackage{setspace}
\usepackage{url}
\usepackage{amsmath}
\usepackage{amssymb}
\usepackage{amsxtra}
\usepackage{algorithmic}
\usepackage{algorithm}
\usepackage{enumerate}
\usepackage{cite}
\usepackage{flushend}

\newtheorem{thm}{Theorem}

\newtheorem{prop}[thm]{Proposition}

\newtheorem{defn}[thm]{Definition}

\newtheorem{rem}{Remark}

\begin{document}

\mainmatter  

\title{Indices of Power in Optimal IDS Default Configuration: Theory and Examples}

\titlerunning{Optimal Default IDS Configuration}

%
%
\author{Quanyan Zhu%
\and Tamer Ba\c{s}ar\thanks{This work was supported in part by the  U.S. Air Force Office of Scientific Research (AFOSR) under grant number  AFOSR MURI FA9550-09-1-0249, and in part by Boeing Company through the Information Trust Institute of the University of Illinois.}}
\authorrunning{Q. Zhu and T. Ba\c{s}ar}

\institute{Coordinated Science Laboratory and \\
Department of Electrical and Computer Engineering, \\
University of Illinois at Urbana Champaign,\\
1308 W. Main St., Urbana, IL, USA, 61801,\\
Email: \{zhu31, basar1\}@illinois.edu}

%
%
%
\maketitle

\begin{abstract}
 Intrusion Detection Systems (IDSs) are becoming essential to protecting modern information infrastructures. The effectiveness of an IDS is directly related to the computational resources at its disposal. However, it is difficult to guarantee especially with an increasing demand of network capacity and rapid proliferation of attacks.
On the other hand, modern intrusions often come as sequences of attacks to reach  some predefined goals. It is therefore critical to identify the best default IDS configuration  to attain the highest possible overall protection within a given resource budget.
This paper proposes a game theory based solution to the problem of optimal signature-based IDS configuration under resource  constraints. We apply the concepts of indices of power, namely, Shapley value and Banzhaf-Coleman index, from cooperative game theory to quantify the influence or contribution of libraries in an IDS with respect to given attack graphs.  Such valuations take  into consideration the knowledge on common attack graphs  and experienced system attacks and are used to  configure an IDS optimally at its default state by solving a knapsack optimization problem. 
\end{abstract}

\section{Introduction}
\label{sec:Introduction}

The issue of optimal IDS configuration and provisioning has always been difficult to deal with, mainly due to the overwhelming number of parameters to tune. 
IDSs are generally shipped with a number of attack detection libraries (also known as categories \cite{roesch1999lisa} or analyzers \cite{paxson1998bsd}) with a considerable set of configuration parameters. The current version of the Snort IDS  \cite{roesch1999lisa}, for example, has approximately 10,000 signature rules located in fifty categories. Each IDS also comes with a \emph{default configuration} to use when no additional information or expertise is available. It is not trivial to determine the optimal configuration of an IDS  because it is essential to understand the quantitative relationship between the wide range of analyzers and tuning parameters. This explains the reason why current IDSs are configured and tuned simply based on a trial-and-error approach.
Although there have been recent approaches, such as in \cite{sinha06raid, vasiliadis08raid, schear08raid}, to optimize IDS resource consumption, we still need to deal with resource constraints and make the best use of an IDS with available resource budgets. On the other hand, most of current computer attacks do not come in  one shot but in several steps, by which attackers can acquire an increasing amount of knowledge and privileges to attack the target system. To describe such multi-stage behaviors, \emph{attack graphs or trees} are commonly used as tools to model security vulnerabilities of a system and all possible sequences of exploits used by intruders.

In this paper, we develop a novel game theory based solution to the problem of optimal default signature-based  IDS configuration under resource limitations. The solution considers   the costs and functionalities of libraries and defender's knowledge on common attack graphs to configure an IDS optimally at its default state.

The contribution of this paper can be summarized as follows. We introduce the concept of \emph{detectability} of an attack sequence with respect to a given set of IDS libraries and devise metrics to measure the detectability and the efficacy of detection. From a game theoretical perspective, we view a configuration as a coalition among libraries and apply the indices of power, namely, {\it Shapley value} and {\it Banzhaf-Coleman index}, to rank the overall importance of a library for the purpose of intrusion detection, which can be used in a knapsack problem for finding the optimal default configuration. In addition, we extend our results to general attack graphs based on multilinear extension and propose a scheme to approximate the indices of power when the number of libraries is large.

The rest of the paper is organized as follows. In the next section, we summarize some recent related work on IDS configuration and cooperative games. In Section \ref{sec:AttackerAndDetector}, we define the important notion of detectability and establish a mathematical model for attackers and detectors. In Section \ref{sec:model}, we formulate a cooperative game framework to evaluate the indices of power for a given attack sequence. In Section \ref{MLE}, we introduce multilinear extension as a general framework and an approximation technique to evaluate the indices of power. Finally, in Section \ref{sec:Conclusion}, we conclude the paper.

%
%
%


\section{Related Work}
\label{sec:RelatedWork}


%
We find a recent growing literature on performance characterization of IDSs in the computer science community. Some of the related work is summarized as follows.

\subsection{IDS Performance Evaluation}
Gaffney et al. in \cite{gaffney01sp} use a decision analysis that integrates and extends Receiver Operating Characteristics (ROCs) to provide an expected cost metric. They demonstrate that the optimal operation point of an IDS depends not only on the system's own ROC curve and quantities such as  the expected rate of false positives, false negatives, and the cost of operation, but also on the degree of hostility of an environment in which the IDS is situated, such as the probability and the type of an intrusion. Hence, the performance evaluation of an IDS has to take into account both the defender's side and attacker's side. 

In \cite{ZTB}, a network security configuration
problem is studied. A nonzero-sum stochastic game is formulated to
capture the interactions among distributed intrusion detection systems in the network
as well as their interactions against exogenous intruders. 
The authors have proposed the notion of security capacity as the
largest achievable payoff to an agent at an equilibrium to yield
performance limits on the network security, and a mathematical programming approach is used to characterize
the equilibrium as well as the feasibility of a given security
target.

Zhu and Ba\c{s}ar in \cite{zhubasar} use a zero-sum stochastic game to capture the dynamic behavior of the defender and the attacker. The transition between different system states depends on the actions taken by the attacker and the defender.  The action of the defender at a given time instant is to choose a set of libraries as its configuration, whereas the action of the attacker is to choose an attack from a set of possible ones. The change of configuration from one instant to the next implies for the defender to either load new libraries or features to the configuration or unload part of the current ones. The actions taken by the attacker at different times constitute a sequence of attacks used by the attacker. An online Q-learning algorithm is used to learn the optimal defense response strategies for the defender based on the samples of outcomes from the game.

In this paper, we address the issues of optimal default configuration, which is complementary to the one addressed in  \cite{zhubasar}. We find an optimal configuration which can serve as an initial or starting profile for dynamic IDS configuration.

To identify important factors for the performance of an IDS is another crucial investigation.
In \cite{schaelicke03raid}, Schaelicke et al. observe several architectural and system parameters that contribute to the effectiveness of an IDS, such as operating system structure, main memory bandwidth and latency as well as the processor micro-architecture. Memory bandwidth and latency are identified as the most significant contributors to sustainable throughput. CPU power is important as well; however, it has been overlooked in the experiments due to the existence of other closely related architectural parameters, such as deep pipelining, level of parallelism, and caching.
%
%

In \cite{dreg08raid}, the authors investigate the prediction of resource consumption based on traffic profile. An interesting result, which we assume to be available in this paper, is that both CPU and memory usage can be predicted with a model linear in the number of connections. Equally important is the confirmation that the factoring of IDS resource usage with per-analyzer and per-connection scaling is a reasonable assumption.
The authors use this finding to build an analyzer selection and configuration tool that estimates resource consumption per analyzer to determine whether a given configuration is feasible or not. The constraint used is a target CPU load below which the load should remain for a predefined percentage of time.
The actual selection of a feasible analyzer set is however left as a manual task for the IDS operator. 
In our work, we propose a more informed and automated way of IDS configuration decision that takes into account the resource utilization per IDS library as well as the expected intrusion context based on experienced attack sequences or graphs.

\subsection{Attack Graphs}
\label{sec:AttackGraphs}

The generation of attack graphs has received considerable attention in the literature 
\cite{swiler01discex, sheyner02sp, jajodia07tan, wang06comcom, lip05tr, Sheyner04fmco, mehta2006raid}.
Sheyner et al. present in \cite{sheyner02sp, Sheyner04fmco} a tool for automatically generating attack graphs and performing different kinds of formal vulnerability analysis on them.
%
%
%
Attack graphs have also been used in intrusion containment. Foo et al. develop in \cite{foo05dsn} the ADEPTS intrusion containment system in the context of E-commerce environments. The system builds a graph of intrusion goals, localizes intrusions, and deploys responses at the appropriate services to allow the system to work with minimum overall performance degradation. The system takes into consideration the financial impact of an attack and derives response actions that go beyond the simple deactivation or isolation of the infected service/host by considering interaction effects among multiple components of the protected environment.
%
%
%
Finally, in \cite{noel08jnsm}, attack graphs are used to derive optimal IDS placement in a network so as to minimize intrusion risk. The authors developed the TVA tool (Topological Vulnerability Analysis), which can be used to model a network and populate it with information regarding vulnerabilities. The tool is claimed to have the ability to avoid state-space explosion through attack graph reduction.
In this paper, we assume that such knowledge of attack graphs is given or has been acquired previously through experience.

\subsection{Game-theoretical Methods}

Game-theoretical methods appear to be an appropriate framework that connects the performance evaluation of an IDS with the attack sequences or graphs on the intruder's side. The concepts in cooperative game theory become natural to study the contribution of each IDS component to the attack sequence, especially when we view a configuration as a coalition among IDS libraries.

Cooperative game theory studies the outcome of a game when coalitions are allowed among multiple players. The concepts of the core and stable sets are regarded as solutions to $N-$person cooperative games. However, the lack of general existence theorem has led game theorists to look for other solution concepts. Currently, indices of power such as {\it Shapley value}, {\it Banzhaf-Coleman index} of power and their multilinear extensions have been widely used in a variety of literature involving resource allocation and estimation of power in a group of decision-making agents. In \cite{Owen95}, examples are given on the application of indices of power in the analysis of presidential election games with a quantitative conclusion that voters in some states are assuming more power than voters in other states in the election. In  \cite{BKZ06}, Shapley value is used to allocate profit in a multi-retailer and a single supplier cooperative game where players can form inventory-pooling coalition. In \cite{GK08},  an efficient measurement allocation in unattended ground sensor networks is suggested based on Shapley values. It is shown that by allocating measurements proportional to the Shapley value, the observability of localizing a target increases. A similar approach was also found to allocate unit start-up costs among electricity consumers, load and retailers. 

It appears that there has been very little work on using game theoretical methods to study IDS configurations. Similar to the problems involving resource allocations and presidential elections, cooperative game theory lends itself naturally also to studying the relations among libraries in an adversarial environment.

\section{Attacker and Detector Model}
\label{sec:AttackerAndDetector}

We let $\mathcal{L}=\{l_1,l_2,\cdots,l_N\}$ denote the set of a finite number of libraries and $\mathcal{L}^*$ denote the set of all the possible subsets of $\mathcal{L}$, with cardinality $|\mathcal{L}^*|=2^N$. We let $F_i \in \mathcal{L}^*, i\in\{1, 2,\cdots,2^N\}$ be a   configuration set of libraries, which is a subset of $\mathcal{L}$. Each library has a cost associated with it, i.e., there is a mapping function $\mathcal{C}:\mathcal{L}\rightarrow \mathbb{R}_+$ that determines the cost of each library $c_i=\mathcal{C}(l_i)$. Assuming the cost of each library is independent of the others \cite{dreg08raid}, we define the cost of a configuration $F_i$ by $C_{F_i}=\mathcal{C}^*(F_i)=\sum_{x\in F_i} \mathcal{C}(x)$, where $\mathcal{C}^*: \mathcal{L}^*\rightarrow \mathbb{R}_+$ is a mapping function of configuration cost.

The attacker, on the other hand, has different types of attacks $a_i$. Let $a_i\in\mathcal{A}$ be a specific action of attack and $\mathcal{A}=\{a_1, a_2,\cdots,a_M \}$ be the set of possible attacks. We define a sequence of attacks $S_i$ to be a tuple of elements of $\mathcal{A}$, and $\mathcal{A}^*$ be the set of all possible sequences of attacks. The order of the elements in  $S_i$ indicates a sequential strategy of an intrusion. Every attack $a_i\in\mathcal{A}$ incurs a damage $d_i$, given by the mapping function $\mathcal{D}:\mathcal{A}\rightarrow \mathbb{R}_+$, i.e., $d_i=\mathcal{D}(a_i), \forall a_i\in\mathcal{A}$. Assuming that the damage caused by a sequence of attacks does not depend on the order of the sequence and the damage by one attack is independent of other attacks, we define the damage caused by an attack sequence $S_i\in \mathcal{A}^*$ by $D_{S_i}=\mathcal{D}^*(S_i)=\sum_{x\in S_i}\mathcal{D}(x)$.

Each library $l_i$ can only effectively detect certain attacks. We define the set $P_{l_i}\subset \mathcal{A}$ as its scope of detection. A library $l_i$ is capable of detecting an attack $a_i$ if and only if $a_i\in P_{l_i}$, otherwise the library $l_i$ is sure to fail to detect. The definition of detectability of a library configuration follows from the scope of detection of each library.

Without losing generality, we can further assume that $\bigcap_{l_i\in\mathcal{L}}P_{l_i}=\emptyset$ because we can always define libraries to have functions that do not overlap with each other. This is particularly true in practice with signature-based libraries.

\begin{defn}\label{detectability}
An attack sequence $S_i$ is detectable by a library configuration $F_i$ if  $S_i \subseteq T_i,$ where $T_i:=\cup_{l_k\in F_i}P_{l_k}$. An attack sequence $S_i$ is undetectable by $F_i$  if $S_i \subseteq \overline{T}_i$, where $\overline{T}_i:=\mathcal{A}\setminus T_i$.
\end{defn}

Based on Definition \ref{detectability}, we can separate an attack sequence $S_i$ into two separate subsequences $S_i^\circ$ and $S_i^\star$, where $S_i^\circ$ is undetectable and $S_i^\star$ is  detectable. These two sequences satisfy the properties that they are mutually exclusive, i.e., $S_i^\circ\cup S_i^\star=S_i$ and $S_i^\circ\cap S_i^\star=\emptyset$.

An example is given in Fig. \ref{libraryNSequence}, where we have a sequence of five attacks $a_1, a_2, a_3, a_4$ and $a_5$. Detection library $l_1$ can be used to monitor $a_1$ and $a_2$ effectively whereas libraries $l_2$ and $l_3$ can only detect $a_3$ and $a_4$ respectively. However, $a_5$ is alien to the detection system and no library can be used to detect $a_5$. An IDS will rely on the successful detection of earlier known  attack stages to prevent the last unknown one.

\begin{figure}
\begin{center}
  \centerline{
  \includegraphics[width=3in]{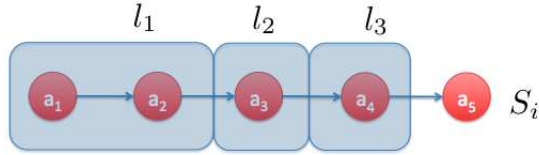}}
  \caption{A library configuration $F_i$ that consists of libraries $l_1, l_2, l_3$ is used to detect an attack sequence $a_1 \rightarrow a_5$.}\label{libraryNSequence}
\end{center}
\end{figure}

The definition of detectability assumes that each library can detect a certain signature or anomaly-based attack with success once it is loaded. However, due to many practical reasons such as delay and mutations of attacks, we can only successfully detect with some true positive (TP) rate.  We use $\alpha^P_{ij}$ to denote the probability of successful detection of an attack $a_i\in\mathcal{A}$ using library $l_{j}$ and by definition $\alpha^P_{ij}=0$ for $i$ not in $P_{l_j}$.
The probability of undetected attacks when attacks occur, or the false negative (FN) rate,  is thus given by $\alpha^N_{ij}=1-\alpha_{ij}$.

We also provide a metric that measures the detectability of an attack sequence $S_i$ with respect to a configuration $F_j$, and the efficiency of detection for $S_i$. 

\begin{defn}
Let function $v:\mathcal{A}^*\rightarrow \mathbb{R}$ be a value function defined on attacker's set of sequences, satisfying 
\begin{equation}\label{axiom1}
v(A_1\cup A_2)\leq v(A_1)+v(A_2),
\end{equation}
where $A_1, A_2\in\mathcal{A}^*$. Given a library component $l_j$, its coverage $P_j$, and an attack sequence $S_i$, we define detection effectiveness, $\eta_{ij}\in [0,1]$, as follows:
\begin{equation}
\label{etaD}\eta_{ij}:=\frac{v(S_i\cap P_j)}{v(S_i)}.
\end{equation}
\end{defn}
\begin{rem}
One simple choice of $v$ is the cardinality of a set, i.e., $v(S_i)=\textrm{card}(S_i)$. We can also have more complicated value functions; for example, we may have more weights on particular important attacks or final attacks.
\end{rem}

Given a configuration $F$ which consists of a finite number of libraries, we use definition in (\ref{detectability}) to define detectability of an IDS configuration with respect to an attack sequence $S_i$ as follows.
\begin{equation}\label{detectability}
\eta_i:=\sum_{l_k\in F}\eta_{ik}=\frac{v(S_i\cap T)}{v(S_i)}=\frac{\sum_{l_k\in F}v(S_i\cap P_k)}{v(S_i)}.
\end{equation}
Using the concepts of TP/FN, we can weight our definition of detectability in (\ref{detectability}) by true positive rate $\alpha^P_{ij}$. Thus we have the definition of weighted detectability as follows.

\begin{defn}
Given a configuration $F$ and attack sequence $S_i$, $\alpha^P$-weighted detectability  is defined as 
\begin{equation}
\eta_i^\alpha=\sum_{l_k\in F}\alpha_{ik}^P\eta_{ik}.
\end{equation}
\end{defn}

The notion of detectability  shows the effectiveness of detection configuration $F_j$ with respect to attack sequence $S_i$.  We will later  use detectability as a metric to optimize the performance of an IDS since higher detectability yields better detection results. 

On the other hand, we can also describe efficiency of a detection using value function $v$. 
\begin{defn}
Given an attack sequence $S_i$ and configuration $F$, we let $\zeta_i$ describe the efficiency of detection
\begin{equation}\label{zeta}
\zeta_{i}=\frac{v(S_i\cap T)}{v(T)}=\sum_{l_k\in F}\frac{v(S_i\cap P_k)}{v(T)},
\end{equation}
and let $\zeta_{i}^\alpha$ denote the weighted detection efficiency given by 
\begin{equation}\label{zetaw}
\zeta_i^\alpha=\sum_{l_k\in F}\alpha_{ik}^P\frac{v(S_i\cap P_k)}{v(T)}.
\end{equation}
where $T$ is a coverage of configuration $F$.
\end{defn}

\begin{prop}
With (\ref{axiom1}) and the metrics $\eta_{i}$ and $\zeta_{i}$ defined in (\ref{detectability}) and (\ref{zeta}) respectively, we have the following relation between these two metrics:
\begin{equation}\label{RELetaDetaE}
\frac{1}{\eta_i}+\frac{1}{\zeta_{i}} \leq 1, \forall i \in \mathcal{A}^*.
\end{equation} 
\end{prop}

\begin{proof}
Using the definitions in (\ref{detectability}), (\ref{zeta}) and (\ref{axiom1}), we arrive at
\begin{eqnarray}
\frac{1}{\eta_{i}}+\frac{1}{\zeta_{i}}&=&\frac{v(S_i)+v(T)}{v(S_i\cap T)}\leq\frac{v(S_i\cup T)}{v(S_i\cap T)}\leq 1.
\end{eqnarray}
\end{proof}

The inequality (\ref{RELetaDetaE}) provides a fundamental tradeoff relationship between detectability and efficiency.

\section{Cooperative Game Model}
\label{sec:model}
\label{model}

In this section, we review essential concepts of indices of power and use them in the context of optimal default IDS configuration. We can view each possible configuration as a coalition of different libraries and hence each library can be associated with an index of power, signifying its contribution to the detection of intrusions.

We introduce the notion of $\omega$-effective detection to quantity the goal of intrusion detection.

\begin{defn}
We call a configuration $F$ of an IDS $\omega-$effective for attack sequence $S_i$ if the detectability  does not fall below $\omega$, that is $\eta_i\geq \omega$, and $(\omega,\alpha)-$effective if  the weighted detectability  does not fall below $\omega$, that is $\eta_i^\alpha\geq\omega$.
\end{defn}

Parameter $\omega$ is a level of intrusion detection performance an IDS wants to achieve. We call a configuration goal achieving if it is  $(\omega,\alpha)$-effective, and unsatisfactory otherwise.

\subsection{Shapley Values and B-C index}

To describe an $N$-person cooperative game using game-theoretical language, we let $\mathcal{L}$ be the set of the players, and any subset of $\mathcal{L}$, or a configuration $F\in\mathcal{L}^*$, be a coalition. We let $f: \mathcal{L}^*\rightarrow \{0, 1\}$ be a characteristic function of the game having basic properties that 
\begin{enumerate}[(P1)]
\item $f(\emptyset)=0$; 
\item $f(F_1\cup F_2)\geq f(F_1)+f(F_2)$, $F_1, F_2\in\mathcal{L}^*$ and $F_1\cap F_2=\emptyset$;
\item $f(\{i\})=0$, for all $i\in\mathcal{L}$.
\item $f(\mathcal{L})=1$. 
\end{enumerate}
By having the characteristic function $f$ taking values $0$ and $1$ only, we have defined a $simple$ game.  The value $1$ from the mapping $f$ of a coalition or configuration means a winning or goal achieving library configuration, whereas the value $0$ yields a non-winning or unsatisfactory library configuration.

A carrier of the cooperative game is a library which does not contribute to any configurations to attain the goal of detection. Mathematically, a carrier is a coalition, $F_c\in\mathcal{L}^*$, such that $f(F)=f(F\cap F_c)$, for all $F$. We can always remove from our library list those dummy libraries  which do not contribute or disregard them in our cooperative game when they are found to be carriers.

The Shapley value of the $i$-th library $l_i$ is given by $\phi_i$, for all $l_i\in\mathcal{L}$. 
\begin{equation}\label{shapley}
\phi_i=\sum_{R\subset  \mathcal{L}}\frac{(r-1)!(N-r)!}{N!}\left[f(R)-f(R-\{l_i\})\right]
\end{equation}
The Shapley value $\phi_i$ given in (\ref{shapley}) evaluates the contribution of each library toward achieving $\omega-$effective detection. Since the characteristic mapping $f$ only takes value in 0 and 1, Shapley value can be further simplified into 
\begin{equation}\label{shapley2}
\phi_i=\sum_{R'\subset\mathcal{L}}\frac{(r-1)!(N-r)!}{N!},
\end{equation}
where, for a given $l_i$,  $R'$ is the winning coalition such that the configuration can achieve $\omega-$effectiveness with $l_i\in R'$, whereas $R'-\{l_i\}$ fails to achieve the goal. With a smaller scale problem, Shapley value is relatively easy to compute. However, in large problems, the evaluation of the weights can create computational overhead and the complexity increases exponentially with the size of the library. An easier index of power to compute is the Banzhaf-Coleman index of power, or B-C index, which depends on counting the number of swings, i.e., number of coalitions or configurations that wins when $l_i$ is included but loses when is not.

\begin{defn}
(B-C Index, \cite{Owen95}) The normalized Banzhaf-Coleman index $\beta_i, \forall l_i\in\mathcal{L}$ is given by 
\begin{equation}
\beta_i=\frac{\theta_i}{\sum_{j=1}^N\theta_j},
\end{equation}
where $\theta_i$ is the number of swings for $l_i$; a swing for  $l_i\in\mathcal{L}$ is a set $R\subset\mathcal{L}$ such that $R$ is a goal-achieving configuration if $l_i\in R$, and $R-\{l_i\}$ is not.
\end{defn}

Shapley value and B-C index are closely related. They can both be evaluated by multilinear extension (see Section \ref{MLE}). The difference lies in the weighting coefficients used. In Shapley value, the weights are varied according to the coalition size, whereas in the B-C index, the weights are all equal.

\subsection{An Example}\label{example1}

Suppose we are given an attack sequence $S_i$ as depicted in Fig. \ref{libraryNSequence}, where we have five attack actions ordered by $a_1 \rightarrow a _2 \rightarrow a_3 \rightarrow a_4 \rightarrow a_5$. There are 3 libraries and the sets $P_{l_i}, i=1,2,3$ are given as follows: $P_{l_1}=\{l_1, l_2\}$,    $P_{l_2}=\{l_3\}$,  $P_{l_3}=\{l_4\}$. It is obvious that the sequence $S_i$ can only be partially detected as $a_5$ is alien to the existing libraries of the IDS system. Suppose that each library has TP rates equal to 1 and $v$ is the cardinality of the set.  We can use Shapley value and B-C index to  quantify  the contribution to the detection of the sequence $S_i$. Let $\omega=3/5$. The set of swings for $l_1$, $l_2$ and $l_3$ are $\{(l_1,l_2),(l_1,l_3), (l_1,l_2,l_3)\},\{(l_1,l_2)\}$, and $\{(l_1,l_3)\}$, respectively. The Shapley values are thus given by $\phi_1=\frac{1}{3}$, $\phi_2=\frac{1}{6}$, and $\phi_3=\frac{1}{6}$; and the B-C indices are thus $\beta_1=\frac{3}{5}$, $\beta_2=\frac{1}{5}$, and $\beta_3=\frac{1}{5}$.
To achieve $\omega =\frac{3}{5}$ level of detection, $l_1$ is most important and $l_2$ and $l_3$ are equally important. Therefore, in terms of the priority of loading libraries, $l_1$ should be placed first and then one should consider $l_2$ and $l_3$. Such evaluation via Shapley value and BC-index is useful for IDS system to assess the influence of each library and make decisions on which libraries to load initially when cost constraints are present.

 \section{Multiple Attack Sequences and Multinear Extension}\label{MLE}
 
 In section \ref{model}, we introduced a cooperative game and proposed the concept of $\omega-$effectiveness to determine the winning or losing coalitions for a given known sequence. In this section, we extend this framework to deal with multiple cooperative games with respect to different sequences in an attack graph. We look at multilinear extensions in this section for two reasons: one is that it can be used to approximate the Shapley value when the library size grows; and the other reason is that it is a general framework that can yield B-C index.
 
 \subsection{Multilinear Extension (MLE)}
 A multilinear extension is a continuous function that can be used to evaluate Shapley value and B-C index as special cases. We let each library $l_i\in\mathcal{L}$ be associated with a continuous variable $x_i\in [0,1]$, and $f_j$ be a characteristic function for detecting particular attack sequence $S_j\in\mathcal{A}^*$. We now introduce the  multilinear extension for an attack sequence $S_j$, denoted by $h_j$, as follows:
  
\begin{defn}
The multilinear extension of the cooperative game with characteristic function $f_j$ is a function $h_j: [0,1]^N\rightarrow \mathbb{R}_{++}$ given by
 \begin{equation}\label{MLEdefn}
 h_j(x_1, x_2, \cdots, x_N)=\sum_{R\subset \mathcal{L}} \left\{ \prod_{l_i\in R} x_i \prod_{l_i\in \mathcal{L}-R}(1-x_i) \right\}f_j(R).
 \end{equation}
 \end{defn}
 
 The function $h_j$ can be used to evaluate Shapley's value by (\ref{MLEsh}) below and B-C index by (\ref{MLEBC}) below.
\begin{equation}\label{MLEsh}
\phi_{ij}=\int_0^1\frac{\partial h_j(x_1, x_2, \cdots, x_N)}{\partial x_i}\Bigg|_{x_1=t,x_2=t,\cdots, x_N=t} dt,
\end{equation}
\begin{equation}\label{MLEBC}
\beta_{ij}=\frac{\partial h_j(x_1, x_2, \cdots, x_N)}{\partial x_i}\Bigg|_{x_1=\frac{1}{2},x_2=\frac{1}{2},\cdots, x_N=\frac{1}{2}},
\end{equation}
where $\phi_{ij}$ and $\beta_{ij}$ are  the Shapley value and B-C index of library $l_i$ for detecting sequence $S_j$, respectively. The set $\mathcal{M}$ is a subset of $\mathcal{A}^*$ that models a set of attacks known to detectors.

To aggregate the effect of a library of detecting a set of sequences $\mathcal{M}\subset\mathcal{A}^*$, we define an aggregated MLE $\bar{h}$ as a sum of MLEs over all the sequences, as follows
\begin{equation}\label{AggBC}
\bar{h}=\sum_{S_j\in\mathcal{M}} p_{j} h_{j},
\end{equation}
where $p_j$ is a weight on $h_j$, indicating the frequency of occurrence of the attack sequence $S_i$. It is a normalized parameter that satisfies $p_j\geq0 $ and $\sum_{S_j\in\mathcal{M}}p_j=1.$

\begin{prop}
The Shapley value $\phi_i$ for detecting multiple sequences is given by
\begin{equation}\label{phiShapley}
\phi_i=\sum_{S_j\in\mathcal{M}}p_j\phi_{ij},
\end{equation}
and B-C index $\beta_i$ for detecting multiple sequences is given by
\begin{equation}\label{phiBC}
\beta_i=\sum_{S_j\in\mathcal{M}}p_j\beta_{ij}=\sum_{S_j\in\mathcal{M}}\left(\frac{p_j\theta_{ij}}{\sum_{k=1}^N\theta_{kj}}\right),
\end{equation}
where $\theta_{ij}$ is the number swings for detecting sequence $S_j$.
\end{prop}
\begin{proof}
The proposition can be proved using the linearity of MLE.
\end{proof} 

\subsection{Multilinear Approximation}
As is well known, the multilinear extension in (\ref{MLEdefn}) has a probabilistic interpretation and it can be used to approximate the indices of power when the number of libraries grows large. We can view the variable $x_i\in [0,1]$ as the probability of a library $l_i$ in a random coalition $\mathcal{S}\subset \mathcal{L}$ such that $f_j(\mathcal{S})=1$ when $l_i\in \mathcal{S}$, and $f_j(\mathcal{S})=0$ otherwise. Since the event that a library is in a random coalition is independent from the event of other libraries in a coalition, we have that the probability of forming the random coalition $\mathcal{S}$ as a particular coalition $R$ is given by
$\mathbb{P}(\mathcal{S}=R)=\prod_{l_i\in R} x_i\prod_{l_i\in \mathcal{L}-R}(1-x_i)$. The definition in (\ref{MLEdefn}) can be interpreted as the expectation of $f(\mathcal{S})$, i.e., $h_j(x_1, x_2,\cdots, x_N)=\mathbb{E}(f_j(\mathcal{S}))$.

Let $Z_j$ be a random variable such that 
\begin{equation}\label{Zrand}
Z_j=\left\{
\begin{array}{ll}
\sum_{i\in\mathcal{M}}\eta_{ij}=:\eta_j^S, &\textrm{~if~} l_j\in\mathcal{S}\\
0 & \textrm{~if~} l_j\in\mathcal{L}-\mathcal{S} 
\end{array}\right.
\end{equation}
and let $Y$ be another random variable defined by $Y=\sum_{l_j\in\mathcal{S}}\eta_j^S=\sum_{l_j\in \mathcal{S}, j\neq i} Z_j$. Since $Z_j$'s are independent, $Y$ has the mean and variance
\begin{equation}\label{Ymean}
\mu(Y)=\sum_{j\neq i,l_j\in S}\eta_j^S x_j,
\end{equation}
\begin{equation}\label{Yvar}
\sigma^2(Y)=\sum_{j\neq i, l_j\in S} \eta_{j}^Sx_j(1-x_j).
\end{equation}
Hence, ${h}_i(x_1,\cdots, x_N)$ is the probability that a coalition wins to be $\omega$-effective with respect to a set of sequences $\mathcal{M}$ but loses if $l_i$ is removed from the coalition. From the definition of $\omega$-effectiveness, we can express ${h}_i$ as
\begin{equation}
{h}_i(x_1,x_2, \cdots, x_N)=\mathbb{P}(\omega \leq Y\leq \omega+\eta_i^S)
\end{equation}
When the size of the library grows, the random variable $Y$ can be approximated by a normal random variable $\bar{Y}$, with mean and variance given in (\ref{Ymean}) and (\ref{Yvar}). Hence,
\begin{equation}\label{prox}
{h}_i(x_1,x_2, \cdots, x_N)=\mathbb{P}\left(\omega-\frac{1}{2} \leq \bar{Y}\leq \omega+\eta_i^S-\frac{1}{2}\right).
\end{equation}

The Shapley value can thus be  computed from (\ref{MLEsh}) and (\ref{AggBC}) by
${h}_i(t,t,\cdots,t)$ with the random variable $\bar{Y}$ having the mean and variance $\mu(\bar{Y})=t\sum_{j\neq i}\eta_j^S$ and $\sigma^2(\bar{Y})=t(1-t)\sum_{j\neq i} \eta_{j}^S$, respectively. The B-C value can be approximated by evaluating ${h}_i(t,t,\cdots,t)$ at $t=\frac{1}{2}$ using (\ref{prox}).

\subsection{Optimal Default Configuration}

The indices of power rank the importance of each library from high to low. We can make use of these indices to make a decision on which libraries to load when the system is subject to some cost constraint $C_0$. Toward that end, we arrive at an integer programming problem as follows:
\begin{eqnarray}\label{optConfig}
\max_{\mathbf{z}} & \sum_{l_i\in\mathcal{L}}z_i\phi_i\\
\nonumber \textrm{s.t.}&  \sum_{l_i\in\mathcal{L}}z_ic_i\leq C_0\\
\nonumber & z_i\in\{0,1\}, \forall l_i\in\mathcal{L}
\end{eqnarray}

We can use the B-C index as well in the objective function. The optimization problem (\ref{optConfig}) can be viewed as a knapsack problem \cite{MS90}. The knapsack problem is well-known to be NP-complete. However, there is a pseudo-polynomial time algorithm using dynamic programming and a fully polynomial-time approximation scheme, which invokes the pseduo-polynomial time algorithm as a subroutine.

\subsection{An Example}
\label{sec:Example}

In this section, we continue with the example in Section \ref{example1}, but with an extended attack tree depicted in Fig. \ref{attackTree}. The libraries that can be used for detection are $l_1,l_2,l_3$ and $l_4$ whose coverages are $P_{l_1}=\{a_1,a_2\},P_{l_2}=\{a_3,a_7\},P_{l_3}=\{a_4,a_8\}$ and $P_{l_4}=\{a_6\}$, respectively. There are 4 known attack sequences in the attack tree. Let $S_1$ be the sequence $a_1\rightarrow a_2 \rightarrow a_3 \rightarrow a_4 \rightarrow a_5$; $S_2$ denote $a_1 \rightarrow a_2 \rightarrow a_6$; $S_3$ denote $ a_1 \rightarrow a_2  \rightarrow a_3  \rightarrow a_7$; and $S_4$ be $ a_1 \rightarrow a_2  \rightarrow a_3  \rightarrow a_4  \rightarrow a_8$.

\begin{figure}[here]
\begin{center}
  \center
  \includegraphics[width=3in]{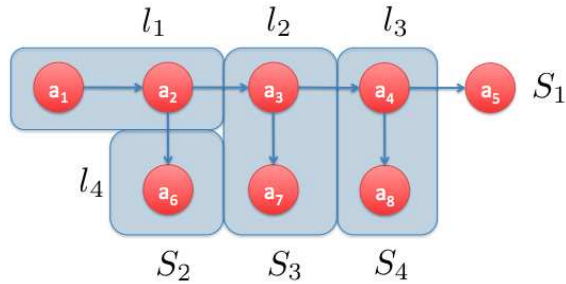}\\
  \caption{Attack tree with attacks $a_j, j=1, 2, \cdots, 8,$ and libraries $l_i, i=1, 2, 3, 4,$ that are used to detect attacks.}\label{attackTree}
\end{center}
\end{figure}

The Shapley values and B-C indices are summarized in the Table 1 and Table 2, respectively. Suppose each library is equally expensive with 1 unit per library. With the capacity constraint being 2 units, we can load library $l_1$ and $l_2$ to optimize the default library. This choice is intuitively plausible because $l_1$ and $l_2$ covers the major routes in the attack tree. $l_4$ does not contribute to the result of detection as much as other libraries, and when $\omega = \frac{3}{5}$, the impact of $l_4$ becomes negligible when $l_1$ is used. When the size of the tree grows, we need to evaluate indices of power in an automated fashion and use a polynomial-time algorithm to find the solution to the knapsack problem (\ref{optConfig}).

\begin{table}[htdp]
\begin{center}
\begin{tabular}{|c || c | c | c | c |}
\toprule
\ \ Sequences \ \ &$\ \ \ l_1\ \ \ $ &$\ \ \ l_2\ \ \ $  &$\ \ \ l_3\ \ \ $  & $\ \ \ l_4\ \ \ $  \\
\midrule
$S_1$&$ \frac{1}{3}$  &$ \frac{1}{6}$   &$ \frac{1}{3} $  & 0 \\
\hline
$S_2$&1 &0  &0  & 0 \\
\hline
$S_3$&$ \frac{1}{2}$  &$ \frac{1}{2}$   &0  & 0 \\
\hline
$S_4$&$ \frac{1}{3}$  &$ \frac{1}{3}$   &$ \frac{1}{3}$   & 0 \\
\bottomrule
\end{tabular}
\end{center}
\label{ATSV1}
\caption{Shapley value for the attack tree}
\end{table}%

{\Huge
\begin{table}[htdp]
\begin{center}
\begin{tabular}{|c||c|c|c|c|}
\toprule
\ \ Sequences \ \ &$\ \ \ l_1\ \ \ $ &$\ \ \ l_2\ \ \ $  &$\ \ \ l_3\ \ \ $  & $\ \ \ l_4\ \ \ $  \\
\midrule
$S_1$&$\frac{3}{5}$ &$\frac{1}{5}$  &$\frac{1}{5}$  & 0 \\
\hline
$ S_2$ &1 &0  &0  & 0 \\
\hline
$ S_3$ &$ \frac{1}{2}$  &$ \frac{1}{2} $  &0  & 0 \\
\hline
$ S_4$ &$ \frac{1}{3}$  &$ \frac{1}{3}$   &$ \frac{1}{3}$   & 0 \\
\bottomrule
\end{tabular}
\end{center}
\label{ATBC2}
\caption{B-C index for the attack tree}
\end{table}}%

\section{Conclusion}
\label{sec:Conclusion}

In this paper, we have adopted a cooperative game approach to study the influence of each library when forming a configuration or a coalition to detect intrusions according to some known attack graphs. We have used the game approach to connect the detector and the attacker, and developed novel notions of detectability and efficacy of detection. The paper has described the applications of Shapley value and B-C index in a combinatorial knapsack optimization problem, which gives rise to an optimal configuration under the resource and cost constraints. The multilinear extension offers a technique to generalize the two indices discussed in the paper and, in addition, offers an approach to estimate these values when the number of libraries grows.

\section*{Acknowledgement}
The first author acknowledges the fruitful discussions with members of the security research group at the University of Waterloo, in particular, Dr. I. Aib and Prof. R. Boutaba.
\bibliographystyle{splncs03}
\bibliography{bibfile}
\end{document}